\documentclass[11pt]{article}
\usepackage{array,amsmath,amssymb,amsthm,color,vmargin,overpic,amstext,mathtools,bbm,hyperref,tikz}
\hypersetup{
	colorlinks=true,
	linkcolor=blue,
	citecolor=blue,
	urlcolor=blue
}
\usepackage{amsfonts,amscd,bm,float}
\usepackage{tikz,ytableau,caption,subfig,tabularx}
\usetikzlibrary{positioning}
\usepackage{graphicx}

\usepackage[T1]{fontenc}
\usepackage[utf8]{inputenc}
\usepackage[english]{babel}

\usepackage{amsmath,amssymb,amsfonts,amsthm,mathtools}
\usepackage{mathrsfs}

\usepackage[square,numbers,sort&compress]{natbib}
\usepackage{xcolor}
\usepackage{caption}
\usepackage{tikz}
\usetikzlibrary{arrows.meta,calc}

\newtheorem{Theorem}{Theorem}[section]
\newtheorem{Corollary}[Theorem]{Corollary}
\newtheorem{remark}[Theorem]{Remark}
\newtheorem{prop}[Theorem]{Proposition}
\newtheorem{lemma}[Theorem]{Lemma}

\numberwithin{equation}{section}
\title{Higher order derivative moments of CUE characteristic polynomials and the Riemann zeta function}

\author{
  Alexander Grover\thanks{School of Mathematics, University of Bristol, Woodland Road, Bristol, BS8 1TW, UK. Emails: \texttt{s.grover@bristol.ac.uk}, \texttt{f.mezzadri@bristol.ac.uk}} \qquad
  Francesco Mezzadri\footnotemark[1] \qquad
  Nick Simm\thanks{Department of Mathematics, University of Sussex, Falmer, Brighton, BN1 9RH UK. Email: \texttt{n.j.simm@sussex.ac.uk}}
}

\date{}

\begin{document}
\maketitle
\begin{abstract}
We use random matrix theory for the Circular Unitary Ensemble (CUE) to study moments of derivatives of the Riemann zeta function shifted a small distance from the critical line. The corresponding CUE moments are studied in the limit of large matrix size in two regimes: when the spectral parameter is (1) suitably far inside the unit disc, and (2) at a small distance from the unit circle. In case (1), we obtain an asymptotic formula as a combinatorial sum over contingency tables, while in case (2) we obtain a sum over certain determinants with multiplicative coefficients given by Kostka numbers. The latter result is also valid exactly on the unit circle. Then, we consider the analogous problem for mean values of derivatives of the zeta function with suitable shifts. Assuming the Lindel\"of hypothesis, we show that this mean value gives rise to the same sum over contingency tables obtained in the CUE. For sufficiently low-order moments, we establish this result unconditionally.
\end{abstract}

\section{Introduction and results}
The utility of random matrix theory in studying properties of the Riemann zeta function dates back to a famous meeting between Dyson and Montgomery in 1972. Specifically, they identified the equivalence between the pair correlation of zeros of $\zeta(s)$ and those of local eigenvalue statistics of large random matrices \cite{DysonIII,montgomerypair}. This correspondence has since been the subject of extensive numerical and theoretical investigations. Keating and Snaith~\cite{KeatingSnaith} further deepened this connection by showing that mean values of the Riemann zeta function can be modelled by characteristic polynomials of matrices from the Circular Unitary Ensemble (CUE). Soon after, 
in his thesis Hughes \cite{ChrisPhd} obtained an asymptotic formula for the joint integer moments of derivatives of CUE characteristic polynomials and used it to conjecture the corresponding moments of the Riemann zeta function. This led to extensive research on moments of the derivative
\cite{Alvarez25, B19,Forrester2022Joint,Ba19,CRS06,Dehaye2,Dehaye,ForresterPainleve,Winn}. More recently, generalizations to higher order derivatives \cite{jointmoments1,jointmoments2}, relations to Painlev\'{e} equations and non-integer moments of derivatives \cite{AKW22, bothnerwei, AGKW26} have been investigated.

We recall that the CUE is the probability space obtained by equipping the group of $N \times N$ unitary matrices $\mathrm{U}(N)$ with the normalized Haar measure. We define the characteristic polynomial as $\Lambda_{N}(z) = \det(1-zU^{\dagger})$ where $U^{\dagger}$ is the conjugate transpose of $U \in \mathrm{U}(N)$. The joint moments we consider in this paper are defined as
\begin{equation}\label{defmoms}
M_{\mu,\nu}(z,N) := \mathbb{E}\left[\prod_{i=1}^{\ell(\mu)}\Lambda_N^{(\mu_i)}(z)\prod_{j=1}^{\ell(\nu)}\overline{\Lambda_N^{(\nu_j)}(z)}\right],
\end{equation}
where the expectation is taken with respect to Haar measure. Here, $\mu=(\mu_1,\dots,\mu_{\ell(\mu)})$ and $\nu=(\nu_1,\dots,\nu_{\ell(\nu)})$ are finite lists of non-negative integers indexing the joint moment \eqref{defmoms}. The earlier cited works obtain asymptotics of such averages for various choices of $\mu$ and $\nu$, exclusively under the assumption that $|z|=1$. Here, our main results provide an asymptotic formula at points sufficiently far inside the unit disc (Theorem \ref{RMTTheorem}) and another result at a microscopic distance from the unit circle, which also covers the case $|z|=1$ (Theorem \ref{thm:micro2}).

For a finite list $\mu$, we define the weight $|\mu| = \sum_{j}\mu_{j}$ and the factorial $\mu! = \prod_{j}\mu_{j}!$. Let $\mathcal{M}_{\mu,\nu}$ denote the set of $\ell(\mu) \times \ell(\nu)$ matrices over the non-negative integers with row sums $\mu$ and column sums $\nu$. We use the notation $\mathcal{M}_{\mu,.}$ (or $\mathcal{M}_{.,\nu}$) to indicate that the columns (or rows) have no constraints.
\begin{Theorem}\label{RMTTheorem}
Fix $\alpha \in (0,1)$. The following asymptotics are uniform for $z$ varying in the disc $|z|<1-N^{-\alpha}$, as $N \to \infty$
\begin{equation}\label{eqn:jacobi}
    M_{\mu,\nu}(z,N) = \frac{ \mu! \nu! }{(1-|z|^2)^{\ell(\mu)\ell(\nu) + |\mu| +|\nu|}}\sum_{\substack{ Q \in \mathcal{M}_{\mu, \cdot} \\ R \in \mathcal{M}_{\cdot, \nu} }} \prod_{i=1}^{\ell(\mu)}\prod_{j=1}^{\ell(\nu)} p_{Q_{ij}, R_{ij}}(z,\bar z)+ O(e^{-N^{\delta}}),
\end{equation}
for some $\delta>0$, where
\begin{equation}
    p_{n,m}(z,\bar{z}) := \sum_{k=0}^{\min(n,m)} \binom{n}{k}\binom{m}{k} z^{m-k}\bar{z}^{n-k}. \label{poly}
\end{equation}
\end{Theorem}
We first comment on some limiting cases of the above result and its relation to the existing literature. The structure of Theorem \ref{RMTTheorem} can be compared to results for the CUE due to Diaconis and Gamburd \cite{diaconis2004random}. To see this, consider the particular case $z=0$ in Theorem \ref{RMTTheorem} where \eqref{eqn:jacobi} simplifies further due to the fact that the only non-zero contribution to $p_{Q_{ij},R_{ij}}(0,0)$ occurs when $k=\min(Q_{ij},R_{ij})$ and this forces $Q=R$. It follows that
\begin{equation}\label{eqn:DGequivalence}
\frac{1}{\mu!\nu!}\lim_{N \rightarrow \infty}M_{\mu,\nu}(0,N) = \sum_{\substack{ Q\in\mathcal{M}_{\mu,\cdot}\\R\in\mathcal{M}_{\cdot,\nu}}} \mathbbm{1}_{Q=R} = N_{\mu\nu}, 
\end{equation}
where $N_{\mu\nu}$ is the number of $\ell(\mu) \times \ell(\nu)$ non-negative integer matrices with row sums $\mu$ and column sums $\nu$. The latter integer valued arrays appeared in \cite{diaconis2004random} where they are referred to as magic squares. In particular, they studied the so-called secular coefficients $\{c^{(N)}_{j}\}_{j=0}^{N}$, defined via the characteristic polynomial as $\Lambda_{N}(z) = \sum_{j=0}^{N}c^{(N)}_{j}z^{j}$. Since $\Lambda^{(j)}_{N}(0) = j!c^{(N)}_{j}$, Theorem \ref{RMTTheorem} gives the relation
\begin{equation}
    \lim_{N \to \infty}\mathbb{E}\left[\prod_{i=1}^{\ell(\mu)}c^{(N)}_{\mu_i}\prod_{j=1}^{\ell(\nu)}\overline{c^{(N)}_{\nu_j}}\right] = N_{\mu\nu},
\end{equation}
which is a limiting case of \cite[Theorem 2]{diaconis2004random}.

In the opposite limit where $z \to 1$, the order of the divergence $\ell(\mu)\ell(\nu)+|\mu|+|\nu|$ in Theorem \ref{RMTTheorem} matches the order of $N$ in the asymptotic results for the moments evaluated at $z=1$, obtained the previous work~\cite{B19,yaccine,jointmoments1,jointmoments2}. Based on results in \cite{SW26} for the first derivative, one expects the order of magnitude of such averages to saturate at the microscopic scaling $z=1-\frac{c}{N}$ for a constant $c$, i.e.\ at a distance from the unit circle on the order of the mean separation between adjacent eigenvalues. Here, one obtains qualitatively different expressions for the leading order term.

Recall that the \textit{Kostka number} $K_{\lambda \mu}$ is the number of semistandard Young tableaux of shape $\lambda$ and content $\mu$. See \cite{Macdonald1995} for background on these coefficients. For this to make sense when computing $K_{\lambda \mu}$, we interpret $\mu$ as a partition, meaning that we arrange the entries of $\mu$ to be weakly decreasing and discard any zero entries. We adopt this convention throughout the paper. For $s \in \mathbb{N}$ and partitions $\lambda$ and $\rho$ with at most $s$ parts, define
\begin{equation}
\alpha_{i}(\lambda)\coloneq \lambda_{i}+s-i,
\qquad
\beta_{j}(\rho)\coloneq \rho_{j}+s-j,
\end{equation}
and
\begin{equation}
\lambda!_{(s)} = \prod_{i=1}^{s}\alpha_{i}(\lambda)!\,, \qquad \rho!_{(s)} = \prod_{j=1}^{s}\beta_{j}(\rho)!\,.
\end{equation}
\begin{Theorem}
\label{thm:micro2}
Let $\mu$ and $\nu$ be lists of non-negative integers with $\ell(\mu)=\ell(\nu)=s$, and define
\begin{equation}
z_{N}=1-\frac{c}{N},
\qquad
\tau\coloneq c+\overline{c}.
\end{equation}
Then as $N \rightarrow \infty$, uniformly in compact subsets of $c \in \mathbb{C}$, we have
\begin{equation}
\label{climit}
\lim_{N \to \infty}\frac{M_{\mu,\nu}(z_{N},N)}{N^{|\mu|+|\nu|+s^{2}}}
=
\mu!\nu!
\sum_{\substack{\lambda\vdash |\mu|,\ \ell(\lambda)\le s\\
\rho\vdash |\nu|,\ \ell(\rho)\le s}}
\frac{K_{\lambda\mu}K_{\rho\nu}}{\lambda!_{(s)}\,\rho!_{(s)}}
\det\!\Big(I_{\alpha_{i}(\lambda)+\beta_{j}(\rho)}(\tau)\Big)_{i,j=1}^{s},
\end{equation}
where
\begin{equation}
I_{r}(\tau)\coloneq \int_{0}^{1}x^{r}e^{-\tau x}\,dx. \label{ir}
\end{equation}
In particular, when $c=0$,
\begin{equation}
\lim_{N \to \infty}\frac{M_{\mu,\nu}(1,N)}{N^{|\mu|+|\nu|+s^{2}}}
=
\mu!\nu!
\sum_{\substack{\lambda\vdash |\mu|,\ \ell(\lambda)\le s\\
\rho\vdash |\nu|,\ \ell(\rho)\le s}}
\frac{K_{\lambda\mu}K_{\rho\nu}}{\lambda!_{(s)}\,\rho!_{(s)}}
\det\!\left(
\frac{1}{\lambda_{i}+\rho_{j}+2s-i-j+1}
\right)_{i,j=1}^{s}. \label{c0limit}
\end{equation}
Furthermore, the right-hand sides in \eqref{climit} and \eqref{c0limit} are strictly positive.
\end{Theorem}
As we mentioned earlier, apart from \cite{SW26}, all known asymptotic results for $M_{\mu,\nu}(z,N)$ are for $|z|=1$, which corresponds to the special case $c=0$ in Theorem \ref{thm:micro2}. This was studied in \cite{yaccine} for $\mu=\nu$, where a complicated multiple integral is given for the leading term in the asymptotics. More recently, the works \cite{jointmoments1,jointmoments2}, again for $c=0$, established the asymptotics \eqref{c0limit} for the particular case (in multiplicity form) $\mu=\nu = (n_{1}^{M},n_{2}^{k-M})$, though their leading term is given by a different combinatorial formula. The structure of Theorem \ref{thm:micro2} can be regarded as a natural generalisation of \cite[Theorem 1.6]{SW26}, which corresponds to the first derivative. While being more general, the proof we give is conceptually more straightforward. Regarding the strict positivity in Theorem \ref{thm:micro2}, we emphasise that this property is not always easy to establish, see \cite{CRS06} where this issue is discussed. Without it, one would not be sure that the power of $N$ in \eqref{climit} and \eqref{c0limit} gives the true order of the asymptotics. Finally, we remark that a simultaneous work \cite{AAKP} independently found the connection between higher order derivatives and Kostka numbers underpinning Theorem \ref{thm:micro2}. 

We now discuss a version of Theorem \ref{RMTTheorem} in the context of mean values of the Riemann zeta function. Their study began with the work of Hardy and Littlewood \cite{hardylittlewood}, who showed
\begin{equation}\label{eqn:hardylittle}
\frac{1}{T}\int_{1}^T |\zeta(1/2+it)|^{2}\,dt \sim \log T  \quad \text{as} \quad T \rightarrow \infty.
\end{equation}
Higher-order moments of the form 
\begin{equation}
\frac{1}{T} \int_1^T |\zeta(1/2 + it)|^{2s} \, dt \sim a_s b_s (\log T)^{s^2} \quad \text{as} \quad T \rightarrow \infty, \label{highermomsks}
\end{equation}
were later studied in a rich history starting with Ingham \cite{ingham} who showed $a_2b_2=\frac{1}{2\pi^2}$. Conjectures and bounds were subsequently made \cite{conrey1984mean,conrey2001high} for the third and fourth moments ($b_3=\frac{42}{9!}$ and $b_4=\frac{24024}{16!}$), where the arithmetic contribution to \eqref{highermomsks} $a_s$ was known
\begin{equation}\label{eqn:nonuniversalpart}
a_s = \prod_{p} \left\{\left(1 - \frac{1}{p}\right)^{s^2}\left(\sum_{m=0}^{\infty}\left(\frac{\Gamma(s+m)}{m!\,\Gamma(s)}\right)^2 p^{-m}\right)\right\}.
\end{equation}
It was not until the aforementioned paper of Keating and Snaith \cite{KeatingSnaith} that the complete structure of $b_s$ was conjectured to arise from moments of large-$N$ CUE random matrices. They predicted
\begin{equation}\label{eqn:Barnesg}
b_s=\lim_{N \rightarrow \infty }N^{-s^2}\mathbb{E}\left[\left|\Lambda_N\left(1\right)\right|^{2s}\right]=\frac{G^2(s+1)}{G(2s+1)},
\end{equation}
where $G$ is the Barnes G-function. Here the moments are evaluated at $z=1$, since by rotational invariance of the CUE, it is equivalent to taking any point on the unit circle. They also provided a complete finite $N$ description of the above moments in the CUE. We remark that to our knowledge, there are no number-theoretical results establishing \eqref{highermomsks} for $s > 2$.

When comparing to the random matrix results with $|z|<1$ fixed, the idea is to consider moments of the Riemann zeta function shifted a fixed distance from the critical line. These are much more tractable analytically than the conjectured asymptotics \eqref{highermomsks}. After the limit $T \to \infty$ is taken for a fixed $\sigma>\frac{1}{2}$, we then let $\sigma \to \frac{1}{2}$ from above. We want to compare this with Theorem \ref{RMTTheorem} as $z \to 1$. Applying Theorem \ref{RMTTheorem} and the binomial sum
    \begin{equation}
        p_{n,m}(1,1) := \sum_{k=0}^{\min(n,m)}\binom{n}{k}\binom{m}{k} = \binom{n+m}{m}, \label{binsum}
    \end{equation}
we obtain, as $|z| \to 1$,
\begin{equation}
\lim_{N \rightarrow \infty}M_{\mu,\nu}(z,N)\sim (1-|z|^2)^{-\ell(\mu)\ell(\nu)-|\mu|-|\nu|}\mu!\nu!\sum_{\substack{ Q\in\mathcal{M}_{\mu,\cdot}\\R\in\mathcal{M}_{\cdot,\nu}}}\prod_{i=1}^{\ell(\mu)}\prod_{j=1}^{\ell(\nu)} \binom{Q_{ij}+R_{ij}}{Q_{ij}}. \label{leading}
\end{equation}
In the next result, assuming the Lindel\"of hypothesis, we show that the above random matrix contribution arises for an appropriate mean value analogue of \eqref{defmoms}.

\begin{prop}\label{thm:ZetaThmLindel}
Assuming the Lindel\"of hypothesis, we have the following asymptotic relation as $\sigma \to \frac{1}{2}$ from above:
\begin{equation}\label{eqn:ZetaMoments}
\lim_{T\rightarrow \infty} \frac{1}{T} \int_{1}^{T}
\prod_{j=1}^{\ell(\mu)}\zeta^{(\mu_j)}(\sigma+it)
\prod_{k=1}^{\ell(\nu)} \overline{\zeta^{(\nu_k)}(\sigma+it)}\,dt
\;\sim\;
\frac{(-1)^{|\mu|+|\nu|}a_{\mu,\nu} h_{\mu,\nu
}}{
  (2\sigma - 1)^{\ell(\mu)\ell(\nu)+|\mu|+|\nu|}
},
\end{equation}
where the arithmetic factor
\begin{equation}\label{eqn:generalarithmetic}
a_{\mu,\nu}=\prod_p \left\{(1-p^{-1})^{\ell(\mu)\ell(\nu)} \sum_{m=0}^{\infty}\left(\frac{\Gamma(\ell(\mu)+m)}{m!\,\Gamma(\ell(\mu))}\right)\left(\frac{\Gamma(\ell(\nu)+m)}{m!\,\Gamma(\ell(\nu))}\right)p^{-m}\right\},
\end{equation}
and the universal factor 
\begin{equation}\label{eqn:Hmunu}
h_{\mu,\nu}=\mu!\nu!\sum_{\substack{ Q\in\mathcal{M}_{\mu,\cdot}\\R\in\mathcal{M}_{\cdot,\nu}}}\prod_{i=1}^{\ell(\mu)}\prod_{j=1}^{\ell(\nu)}\binom{R_{ij}+Q_{ij}}{Q_{ij}}.
\end{equation}
Furthermore, if $\mu$ and $\nu$ satisfy $\ell(\mu), \, \ell(\nu) \leq 2$, the above asymptotics are true unconditionally, i.e.\ without assuming the Lindel\"of hypothesis.
\end{prop}
When $\mu = \nu = (k,k)$, we give a closed form expression for the leading term in the asymptotics which we establish unconditionally.
\begin{Theorem}\label{thm:ZetaThms2}
Let $k$ be a positive integer. Then as $\sigma \rightarrow \frac{1}{2}$, we have
\begin{equation}\label{eqn:zetas2}
\lim_{T \to \infty} \frac{1}{T}
\int_{1}^{T} \bigl|\zeta^{(k)}(\sigma + i t)\bigr|^{4} \, dt
\;\sim\;
\frac{6}{\pi^2}\frac{h_{k
}}{
  (2\sigma - 1)^{4(k+1)}
},
\end{equation}
where
\begin{equation}\label{eqn:hk}
h_{k}=
  (k!)^4(k+1)\left(\binom{4k+3}{2k+2}
  -2 \binom{2 k+1}{k}^2\right).
\end{equation}
\end{Theorem}
Finally, it is reasonable to conjecture that the leading coefficients presented in Theorem \ref{thm:micro2} also arise from appropriately defined mean values of $\zeta$. Specifically, one expects that if the shifts on the left-hand side of \eqref{eqn:ZetaMoments} are taken to be $\sigma = \frac{1}{2}+\frac{c}{\log(T)}$, then the leading order asymptotics will be given by the random matrix result in \eqref{climit}, replacing $N$ with $\log(T)$ and multiplying by the arithmetic factor \eqref{eqn:generalarithmetic}.
\section*{Acknowledgements} 
We thank Nina Snaith and the Heilbronn Institute for Mathematical Research (HIMR) for organising, and respectively funding, the Focused Research Workshop on \textit{Averages of Derivatives of Characteristic Polynomials of Random Unitary Matrices} at which this problem was proposed. The format of this workshop made this collaboration possible. A.~G.\ is grateful to Fei Wei and Andrew Pearce-Crump for helpful conversations surrounding Section~\ref{section3}. We thank Fei Wei for pointing out the strict positivity property of Theorem \ref{thm:micro2}. We are grateful to the authors of \cite{AAKP} for sharing their preprint with us. A.~G.\ acknowledges support from UKRI grant no. 2765681. N.~S. is grateful for support from the Royal Society, grant URF\textbackslash R\textbackslash231028.

\section{Moments of higher order derivatives in the CUE}\label{section2}
In this section we prove Theorems \ref{RMTTheorem} and \ref{thm:micro2}. We first set some notation that we make use of throughout the paper. We abbreviate the lengths of $\mu$ and $\nu$ in \eqref{defmoms} as $K := \ell(\mu)$ and $L := \ell(\nu)$. For the list $\mu$ we work with vectors $\bm{z} = (z_1,\ldots,z_{K}) \in \mathbb{C}^{K}$ and we make use of the differential operators $\mathcal{D}^{\mu}_{\bm{z}} :=   
\prod_{i=1}^{K} \frac{\partial^{\mu_i}}{\partial z_i^{\mu_i}}$. The notation $\mathcal{D}^{\mu}_{\bm{z}}f(\bm z)|_{\bm z = z}$ means that we take the limit $z_{j} \to z$ for all $j$ after applying $\mathcal{D}^{\mu}_{\bm{z}}$. Hence, from \eqref{defmoms}, we have
\begin{equation}
M_{\mu,\nu}(z,N) = \mathcal{D}^{\nu}_{\bm w}\mathcal{D}^{\mu}_{\bm z}\mathbb{E}\bigg[\prod_{i=1}^{K}\Lambda_{N}(z_i)\prod_{j=1}^{L}\overline{\Lambda_{N}(\overline{w_j})}\bigg]\bigg{|}_{\bm z = z, \bm w = \bar z}.
\end{equation}
\begin{lemma}
\label{lem1}
Fix $\alpha \in (0,1)$. The following asymptotics are uniform for $z$ varying in the disc $|z|<1-N^{-\alpha}$, as $N \to \infty$,
\begin{equation}
    M_{\mu,\nu}(z,N) = \mathcal{D}^{\nu}_{\bm w}\mathcal{D}^{\mu}_{\bm z}\prod_{i=1}^{K}\prod_{j=1}^{L}\frac{1}{1-z_iw_j}\Bigg|_{\bm z =z, \bm w = \bar z}+ O(e^{-N^{\delta}}), \label{meso}
\end{equation}
for some $\delta>0$.
\end{lemma}

\begin{proof}
We start with \cite[Theorem 1]{akemann2003characteristic} applied to the CUE, which gives,
\begin{equation}\label{eqn:NF41}
\mathbb{E}\left[ \prod_{i=1}^{K}\Lambda_{N}(z_i)\prod_{j=1}^{K}\overline{\Lambda_{N}(\overline{w_j})} \right]=\frac{\det\left\{\mathcal{K}_N(z_iw_j)\right\}_{i,j=1}^K}{\Delta(\bm{z})\Delta(\bm{w})},
\end{equation}
where $\mathcal{K}_N(r)=\sum_{l=0}^{N+K-1}r^l=\frac{1-r^{N+K}}{1-r}$ and $\Delta(\bm z)$ is the Vandermonde determinant,
\begin{equation}
\Delta(\bm{z}) = \prod_{1 \leq i < j \leq K}(z_j-z_i) = \det\bigg\{z_{i}^{j-1}\bigg\}_{i,j=1}^{K}. \label{van}
\end{equation}
Consider the numerator on the right-hand side of \eqref{eqn:NF41} and first consider $z$ fixed. Since $\mathcal{K}_{N}(r) \to \frac{1}{1-r}$ uniformly on compact subsets of $|r|<1$ as $N \to \infty$, we have
\begin{equation}
\lim_{N\rightarrow \infty}\det\left\{\mathcal{K}_N(z_iw_j)\right\}_{i,j=1}^K= \det\left\{\frac{1}{1-z_iw_{j}}\right\}_{i,j=1}^K = \frac{\Delta(\bm{z})\Delta(\bm{w})}{\prod_{i,j=1}^K(1-z_iw_j)},
\end{equation}
where the second equality above is a consequence of the Cauchy determinant identity, see \cite[Chapter I, Section 4]{Macdonald1995}. Since the convergence is uniform and both sides of \eqref{eqn:NF41} are polynomials, combined with the above explicit evaluation, we obtain for any positive integers $K,L$,
\begin{equation}
\lim_{N\rightarrow\infty}\mathbb{E}\left[\prod_{i=1}^{K}\Lambda_N(z_i)\prod_{j=1}^{L}\overline{\Lambda_{N}(\overline{w}_j)}\right] =  \prod_{i=1}^{K}\prod_{j=1}^{L}\frac{1}{1-z_{i}w_{j}},  \label{lim-char}
\end{equation}
where the convergence is uniform on compact subsets of the variables $\{z_{i}\}_{i=1}^{K}$, $\{w_{j}\}_{j=1}^{L}$ belonging to the open unit disc. The result above for general $K$ and $L$ is obtained from the one for $L=K$ by setting the appropriate variables to zero using $\Lambda_{N}(0)=1$. Combining \eqref{lim-char} with Cauchy's integral formula in several complex variables, we obtain \eqref{meso} for $|z|<1$ fixed. The proof for $|z|<1-N^{-\alpha}$ follows from identical considerations, \cite[Lemma 22]{det}, and the estimate $\mathcal{K}_{N}(z) = \frac{1}{1-z} + O(e^{-N^{1-\alpha}})$.
\end{proof}
\begin{lemma}\label{Cor1}
The following identity holds for any $|z|<1$,
\begin{equation}
\begin{aligned}
&\mathcal{D}^{\nu}_{\bm w}\mathcal{D}^{\mu}_{\bm z}\prod_{i=1}^{K}\prod_{j=1}^{L}\frac{1}{1-z_iw_j}\Bigg|_{\bm z =z, \bm w = \bar z}\\
&=\frac{\mu!\nu! }{(1-|z|^2)^{KL + |\mu| +|\nu|}} [\bm u^\nu][\bm t^\mu]\prod_{i=1}^{K}\prod_{j=1}^{L}\frac{1}{(1-\bar{z}t_i)(1-zu_j) - t_iu_j},\label{utvars}
\end{aligned}
\end{equation}
where the notation $[\bm{u}^\nu]$ denotes coefficient extraction of the monomial $\prod_{j=1}^{L}u_{j}^{\nu_{j}}$, similarly for $[\bm{t}^\mu]$.
\end{lemma}
\begin{proof}
In the left-hand side of \eqref{utvars} we make the change of variables
\begin{equation}
t_i=\frac{z_i-z}{1-|z|^2}, \quad u_j=\frac{w_j-\bar{z}}{1-|z|^2}, \qquad i,j=1,2,\dots
\end{equation}
valid for $|z|<1$. The lemma then follows from
\begin{equation}
\begin{aligned}
1-z_iw_j=(1-|z|^{2})((1-\bar{z}t_i)(1-zu_j) - t_iu_j),
\end{aligned}
\end{equation}
and by expressing the derivatives in terms of the $\bm{u}, \bm{t}$ variables evaluated at $0$, which can be recognised as coefficient extraction. 
\end{proof}
\begin{proof}[Proof of Theorem \ref{RMTTheorem}]
Starting with the right-hand side of \eqref{utvars}, we expand each factor in the product as a geometric series,
\begin{equation}
\begin{aligned}
\frac{1}{(1-\bar{z}t_i)(1-zu_j) - t_iu_j} &= \frac{1}{(1-\bar{z}t_i)(1-zu_j)} \frac{1}{1 - \frac{t_iu_j}{(1-\bar{z}t_i)(1-zu_j)}} \\
&= \sum_{k=0}^\infty \frac{t_i^k u_j^k}{(1-\bar{z}t_i)^{k+1}(1-zu_j)^{k+1}}.
\end{aligned}
\end{equation}
To evaluate the overall coefficient in \eqref{utvars}, we extract the coefficient of $t_i^{Q_{ij}}u_j^{R_{ij}}$ from each factor. Using the negative binomial expansion $[x^n](1-cx)^{-(k+1)} = \binom{n+k}{k}c^n$, we obtain
\begin{equation}
\begin{aligned}
&[t_i^{Q_{ij}}][u_j^{R_{ij}}] \sum_{k=0}^\infty \frac{t_i^k u_j^k}{(1-\bar{z}t_i)^{k+1}(1-zu_j)^{k+1}}\\ 
&= \sum_{k=0}^{\min(Q_{ij}, R_{ij})} \binom{Q_{ij}}{k} \bar{z}^{Q_{ij}-k} \binom{R_{ij}}{k} z^{R_{ij}-k} = p_{Q_{ij},R_{ij}}(z,\overline{z}),
\end{aligned}
\end{equation}
where we recall definition \eqref{poly}. Applying the full monomial coefficient $[\bm{u}^\nu][\bm{t}^\mu]$ fixes the corresponding row and column sums of $Q$ and $R$, leading to the sum over contingency tables
\begin{equation}
\begin{aligned}
[\bm{u}^\nu][ \bm{t}^\mu]\prod_{i=1}^{K}\prod_{j=1}^{L}\frac{1}{1-\bar{z}t_i-zu_j-(1-|z|^2)t_iu_j}&= \sum_{\substack{ Q\in\mathcal{M}_{\mu,\cdot}\\R\in\mathcal{M}_{\cdot,\nu}}}\prod_{i=1}^{K}\prod_{j=1}^{L}p_{Q_{ij},R_{ij}}(z,\bar{z}).
\end{aligned}
\end{equation}
The result then follows upon substitution into Lemmas \ref{Cor1} and \ref{lem1}.
\end{proof}

\begin{lemma}
	\label{lem:len2}
Consider the length $2$ lists $\mu=\nu=(k,k)$. In the limit $|z| \to 1$ we have the following closed form expression,
\begin{equation}
\begin{aligned}
h_{k}&:=\lim_{|z| \to 1}(1 - |z|^2)^{4(k+1)}\lim_{N \rightarrow \infty}\mathbb{E}\left[\bigl|\Lambda_N^{(k)}(|z|)\bigr|^{4}\right]\\
&=
  (k!)^4(k+1)\left(\binom{4k+3}{2k+2}
  -2 \binom{2 k+1}{k}^2\right).
\end{aligned}
\end{equation} 
\end{lemma}
\begin{proof}
With this choice of $\mu$ and $\nu$, Theorem \ref{RMTTheorem}, in particular \eqref{leading} gives
\begin{equation}
h_{k}=(k!)^{4} \sum_{\substack{ Q\in\mathcal{M}_{\mu,\cdot}\\R\in\mathcal{M}_{\cdot,\nu}}}\prod_{i=1}^{2}\prod_{j=1}^{2}\binom{R_{ij}+Q_{ij}}{Q_{ij}}.
\end{equation}
The proof then follows by evaluating the sum over $2 \times 2$ contingency tables. These are parameterised by
\begin{equation}
Q = \begin{pmatrix}
n&& k-n\\
m&& k-m
\end{pmatrix}, \qquad
R = \begin{pmatrix}
k-p&& k-q\\
p&& q
\end{pmatrix},
\end{equation}
for $n,m,p,q\in \{0,\dots,k \}$. We begin by summing over $p,q$,
\begin{equation}
\begin{aligned}
h_{k} &=(k!)^{4}\sum_{Q \in \mathcal{M}_{\mu,\cdot}} \sum_{p = 0}^{k}\sum_{q = 0}^{k}\binom{k-p+Q_{11}}{Q_{11}}\binom{k-q+Q_{12}}{Q_{12}}\binom{p+Q_{21}}{Q_{21}}\binom{q+Q_{22}}{Q_{22}}\\
&=(k!)^{4} \sum_{Q \in \mathcal{M}_{\mu,\cdot}}
\binom{Q_{11}+Q_{21}+k+1}{k}\binom{Q_{12}+Q_{22}+k+1}{k},
\end{aligned}
\end{equation}
where we used a binomial sum similar to \eqref{binsum}. Similarly, the sum over $m,n$ gives 
\begin{equation}
\begin{aligned}
h_{k}
 &=(k!)^4\sum_{n = 0}^{k}\sum_{m = 0}^{k}  \binom{n+m+k+1}{k}\binom{3k-n-m+1}{k}
  \\&=
  (k!)^4(k+1)\left(\binom{4k+3}{2k+2}
  -2 \binom{2 k+1}{k}^2\right).
\end{aligned}
\end{equation} 
\end{proof}
\begin{Theorem}
\label{micro}
	Let $z = 1-\frac{c}{N}$. Without loss of generality suppose that $K \geq L$. Then uniformly in compact subsets of $c \in \mathbb{C}$, we have
	\begin{equation}
		\label{microlimit}
		\lim_{N \to \infty}\frac{M_{\mu,\nu}(z,N)}{ N^{|\mu|+|\nu|+KL}} =  (-1)^{|\mu|+|\nu|}\mathcal{D}^{\mu}_{\bm{a}}\mathcal{D}^{\nu}_{\bm{b}}\frac{\det(\Omega)}{\Delta(\bm a)\Delta(\bm b)}\bigg{|}_{\substack{\bm a = c\\ \bm b = \overline{c}}},
	\end{equation}
	where $\Delta(\bm a)$ is given by \eqref{van},
	\begin{equation}
		\Omega =  
		\begin{pmatrix} 
			1 & a_1 & \dots & a_1^{K-L-1} & I(a_1 + b_1) & \dots & I(a_1 + b_L) \\
			1 & a_2 & \dots & a_2^{K-L-1} & I(a_2 + b_1) & \dots & I(a_2 + b_L) \\
			\vdots & \vdots & \ddots & \vdots & \vdots & \ddots & \vdots \\
			1 & a_K & \dots & a_K^{K-L-1} & I(a_K + b_1) & \dots & I(a_K + b_L)
		\end{pmatrix},
	\end{equation}
	and $I(x) = \frac{1-e^{-x}}{x}$.
\end{Theorem}
\begin{proof}	
The general $K \geq L$ version of \eqref{eqn:NF41} from \cite[Theorem 1]{akemann2003characteristic} reads
\begin{equation}
\mathbb{E}\left[ \prod_{i=1}^K \Lambda_{N}(z_i)\prod_{j=1}^L \overline{\Lambda_{N}(\overline{w_j})} \right]
=
\frac{\det(\Omega_{N})}{\Delta(\bm{z})\Delta(\bm{w})},\label{AV}
\end{equation}
where
\[
\Omega_{N} = 
\begin{pmatrix} 
	z_1^{K-L} \mathcal{K}_{N+L}(z_1, w_1) & \dots & z_1^{K-L} \mathcal{K}_{N+L}(z_1, w_L) & 1 & z_1 & \dots & z_1^{K-L-1} \\
	z_2^{K-L} \mathcal{K}_{N+L}(z_2, w_1) & \dots & z_2^{K-L} \mathcal{K}_{N+L}(z_2, w_L) & 1 & z_2 & \dots & z_2^{K-L-1} \\
	\vdots & \ddots & \vdots & \vdots & \vdots & \ddots & \vdots \\
	z_K^{K-L} \mathcal{K}_{N+L}(z_K, w_1) & \dots & z_K^{K-L} \mathcal{K}_{N+L}(z_K, w_L) & 1 & z_K & \dots & z_K^{K-L-1}
\end{pmatrix},
\]
and
\begin{equation}
\mathcal{K}_{N+L}(z,w) = \sum_{p=0}^{N+L-1} (zw)^p = \frac{(zw)^{N+L} - 1}{zw - 1}. \label{kern}
\end{equation}
Note that $\Omega_{N}$ is size $K \times K$, consisting of a kernel block of size $K \times L$ followed by $K-L$ excess columns. Let
\begin{equation}
	z_{i} = 1-\frac{a_{i}}{N}, \qquad w_{j} = 1-\frac{b_{j}}{N},
\end{equation}
for all $i=1,\ldots,K$ and $j=1,\ldots,L$, where $a_{i}$ and $b_{j}$ are complex numbers independent of $N$. The derivatives transform as
\begin{equation}
\mathcal D_{\bm z}^{\mu}\mathcal D_{\bm w}^{\nu}
=
(-N)^{|\mu|+|\nu|}
\mathcal D_{\bm a}^{\mu}\mathcal D_{\bm b}^{\nu}.
\end{equation}
Then we express the moments as
\begin{equation}
	\begin{split}
M_{\mu,\nu}(z,N) = (-N)^{|\mu|+|\nu|}\mathcal{D}^{\mu}_{\bm{a}}\mathcal{D}^{\nu}_{\bm{b}}\mathbb{E}\left(\prod_{i=1}^K \Lambda_{N}(z_{i})\prod_{j=1}^{L}\overline{\Lambda_{N}(\overline{w}_j)}\right)\bigg{|}_{\substack{\bm a = c\\ \bm b = \overline{c}}}, \label{merging}
\end{split}
\end{equation}
and apply \eqref{AV}. We use
\begin{equation}
	\Delta(\bm z) = (-N)^{-K(K-1)/2}\Delta(\bm a), \qquad \Delta(\bm w) = (-N)^{-L(L-1)/2}\Delta(\bm b),
\end{equation}
and
\begin{equation}
	\mathcal{K}_{N+L}(z_i, w_j) \sim N\frac{1-e^{-a_{i}-b_{j}}}{a_{i}+b_{j}}, \qquad N \to \infty. \label{kernlimit}
\end{equation}
This follows immediately from \eqref{kern}. Hence the entries of $\Omega_{N}$ involving $\mathcal{K}_{N+L}(z_i,w_j)$ will be of size $N$ and as this occupies $L$ columns, it will give a factor $N^{L}$. Consider the matrix composed of the remaining $K-L$ excess columns. If we subtract column 1 from column 2, column 2 will now have size $1/N$. Performing the same operation on all columns gives consecutively higher powers of $1/N$ in each column. Doing this for all $K-L$ columns and factoring the powers out of the determinant gives a total power $N^{-(K-L)(K-L-1)/2}$. Adding all the exponents gives
\begin{equation}
	N^{K(K-1)/2+L(L-1)/2-(K-L)(K-L-1)/2+L} = N^{KL}.
\end{equation}
This gives a total order of $N^{|\mu|+|\nu|+\ell(\mu)\ell(\nu)}$. Taking care of the resulting powers of $-1$ and interchanging the kernel block with the excess columns, we arrive at the formula on the right-hand side of \eqref{microlimit}. To justify exchanging the limit $N \to \infty$ with the merging in \eqref{merging} note that 1) $\frac{\det(\Omega_{N})}{\Delta(\bm z) \Delta(\bm w)}$ is a polynomial, 2) the limit $\frac{\det(\Omega)}{\Delta(\bm a) \Delta(\bm b)}$ in \eqref{microlimit} is an entire function and 3) the convergence in \eqref{kernlimit} is uniform on compact subsets of $a_{i},b_{j} \in \mathbb{C}$. Then the exchange of limits follows from the Weierstrass convergence theorem.
 \end{proof}

 \begin{remark}
  The function $I(x)$ in \eqref{microlimit} can be recognised in terms of the \textit{sine kernel}, $I(2\pi ix) = e^{-\pi ix}\frac{\sin(\pi x)}{\pi x}$, which is expected to arise when one studies limiting correlations of the CUE on the microscopic scale. It will be useful in the following to note that $I^{(r)}(x) = (-1)^{r}I_{r}(x)$ where $I_{r}(x)$ is given by \eqref{ir}.
 \end{remark}
  For $K=L=2$, the determinants in \eqref{microlimit} are $2 \times 2$ and the required merging can be done explicitly. 
 \begin{Corollary}
 \label{cor:4th}
Let $z = 1-\frac{c}{N}$. By rotational invariance it suffices to consider $c\in \mathbb{R}$. Then we obtain the following formula for the $4^{\mathrm{th}}$ moment of the $k^{\mathrm{th}}$ derivative, for any $k \in \mathbb{N}$,
 \begin{equation}
	\begin{split}
		&\lim_{N \to \infty}\frac{\mathbb{E}(|\Lambda^{(k)}_{N}(z)|^{4})}{N^{4+4k}} =\sum_{i,j=0}^{k}\left(C^{(1)}_{ij}(k)I_{4k-i-j+2}(2c)I_{i+j}(2c)\right.\\
        &\left.-C^{(2)}_{ij}(k)I_{3k-i-j+1}(2c)I_{k+i+j+1}(2c)\right),
	\end{split}
\end{equation}
where
\begin{equation}
\begin{split}
C^{(1)}_{ij}(k) &= \frac{(k!)^{2}\binom{k}{i}\binom{k}{j}}{(2k-i+1)!}\frac{(k-i)!(k-j)!}{(2k+1-j)!},\\
C^{(2)}_{ij}(k) &= \frac{(k!)^{2}\binom{k}{i}\binom{k}{j}}{(2k-i+1)!}\frac{j!(k-i)!}{(k+j+1)!}.
\end{split}
\end{equation}
 \end{Corollary}
 \begin{proof}
The limit in \eqref{microlimit} prior to taking the derivatives and merging is
 \begin{equation}
 f(a_1,a_2,b_1,b_2) = \frac{I(a_1+b_1)I(a_2+b_2)-I(a_1+b_2)I(a_2+b_1)}{(a_2-a_1)(b_2-b_1)}. \label{fdef}
 \end{equation}
 Introduce the function
 \begin{equation}
 	G(u,v) = \frac{I(a(u)+b_{1})I(a_{2}+b(v))-I(a(u)+b(v))I(a_{2}+b_{1})}{(a_2-a_1)(b_2-b_1)},
 \end{equation}
  where $a(u) = a_{1}+u(a_{2}-a_{1})$ and $b(v) = b_{1}+v(b_{2}-b_{1})$. Then $G(1,1) = G(1,0)=G(0,0) = 0$ and $G(0,1) = f(a_1,a_2,b_1,b_2)$. By the fundamental theorem of calculus, we have
  \begin{equation}
 \begin{split}
  &f(a_1,a_2,b_1,b_2) = -\int_{0}^{1}du\int_{0}^{1}dv\frac{\partial}{\partial u}\frac{\partial}{\partial v}G(u,v)\\
  &=-\int_{0}^{1}du\int_{0}^{1}dv\,\Big(I'(a(u)+b_{1})I'(a_{2}+b(v))-I''(a(u)+b(v))I(a_{2}+b_{1})\Big). \label{fintrep}
  \end{split}
  \end{equation}
Thus we eliminated the denominator in \eqref{fdef}, facilitating taking derivatives and merging points. Now we differentiate \eqref{fintrep} $k$ times with respect to each variable using the Leibniz rule for differentiating a product $k$ times, and finally merge all the points $a_{1}=a_{2}=b_{1}=b_{2}=c$. This gives the following binomial sum as the integrand
 \begin{equation}
 	\begin{split}
 	&-(1-u)^{k}v^{k}\sum_{i,j=0}^{k}\binom{k}{i}\binom{k}{j}\left(u^{k-i}(1-v)^{j}I_{3k-i-j+1}(2c)I_{k+i+j+1}(2c)\right.\\
    &\left.-u^{k-i}(1-v)^{k-j}I_{4k-i-j+2}(2c)I_{i+j}(2c)\right).
 	\end{split}
 \end{equation}
 The $u$ and $v$ integration can be done explicitly, and we arrive at the stated result.
 \end{proof}
 \begin{remark}
 Taking $c=0$ in Corollary \ref{cor:4th} using \eqref{ir}, we obtain the following formula for moments precisely on the unit circle,
\begin{equation}
\begin{split}
    &\lim_{N \to \infty}\frac{\mathbb{E}(|\Lambda^{(k)}_{N}(1)|^{4})}{N^{4+4k}} = \sum_{i,j=0}^{k}\left(\frac{C^{(1)}_{ij}(k)}{(4k-i-j+3)(i+j+1)}\right.\\
    &\left.-\frac{C^{(2)}_{ij}(k)}{(3k-i-j+2)(k+i+j+2)}\right). \label{z1-4thmom}
    \end{split}
\end{equation}
This agrees with the previously known results, \textit{e.g.}, for the first derivative $k=1$, it reduces to the known value $\frac{61}{10080}$, see e.g.\ \cite[Section 5]{CRS06}. For general $k$ we can also compare \eqref{z1-4thmom} with the corresponding particular case of \cite[Theorem 5]{jointmoments1} that to our understanding involves six iterated summations. We have checked for modest values of $k$ that their expression coincides with \eqref{z1-4thmom}.
\end{remark}
We now present a result that generalizes the exact formula \cite[Theorem 1.5]{SW26} to higher order derivatives. Our approach uses the machinery of symmetric functions.
\begin{lemma}\label{lem:confluentkostka}
Let $s\in \mathbb{N}$ and let $f_{1},\dots,f_{s}$ be analytic in a neighbourhood of $z\in \mathbb{C}$. Let $\mu=(\mu_1,\dots,\mu_s)$ be a list of non-negative integers of length $s$. Then
\begin{equation}
\left.\mathcal{D}^{\mu}_{\mathbf{x}}
\frac{\det\!\big(f_{j}(x_{i})\big)_{i,j=1}^{s}}{\Delta(\mathbf{x})}
\right|_{x_{1}=\cdots=x_{s}=z}
=
(-1)^{\frac{s(s-1)}{2}}\mu!\sum_{\substack{\lambda\vdash |\mu|\\ \ell(\lambda)\le s}}
\frac{K_{\lambda\mu}}{\lambda!_{(s)}}
\det\!\big(f_{j}^{(\alpha_{i}(\lambda))}(z)\big)_{i,j=1}^{s},
\label{k-expan}
\end{equation}
where we recall the notation and conventions outlined above Theorem \ref{thm:micro2}. 
\end{lemma}

\begin{proof}
By \cite[Appendix B]{B00}, we have
\begin{equation}\label{eqn:Intermediate}
\det\!\big(f_{j}(z+y_{i})\big)_{i,j=1}^{s}
=
\sum_{0\le n_{1}<\cdots<n_{s}}
\frac{\det\!\big(f_{j}^{(n_{i})}(z)\big)_{i,j=1}^{s}}{n_{1}!\cdots n_{s}!}
\det\!\big(y_{i}^{n_{j}}\big)_{i,j=1}^{s}.
\end{equation}
We now look to write this as a sum over Young diagrams by reindexing this sum to be defined over the set of weakly decreasing partitions, in order to make use of standard tools of symmetric function theory. Given we are summing over strictly increasing sequences $0\le n_{1}<\cdots<n_{s}$, these may be uniquely reindexed by writing:
\begin{equation}
n_{i}=\lambda_{s+1-i}+i-1,
\end{equation}
for a partition $\lambda$ with at most $s$ parts. We can then form a strictly decreasing sequence by reversing the list,
\begin{equation}
\{\lambda_{i}+s-i\}_{i=1}^{s}.
\end{equation}
After re-indexing in this way, we obtain
\begin{equation}
\det\!\big(f_{j}^{(n_{i})}(z)\big)_{i,j=1}^{s}=(-1)^{\frac{s(s-1)}{2}}\det\!\big(f_{j}^{(\alpha_{i}(\lambda))}(z)\big)_{i,j=1}^{s},
\end{equation}
and
\begin{equation}
\det\!\big(y_{i}^{n_{j}}\big)_{i,j=1}^{s}
=
(-1)^{\frac{s(s-1)}{2}}\det\!\big(y_{i}^{\lambda_{j}+s-j}\big)_{i,j=1}^{s}
=
\Delta(\mathbf{y})\,s_{\lambda}(\mathbf{y}),
\end{equation}
where $s_{\lambda}$ is the Schur polynomial. Moreover, $n_{1}!\cdots n_{s}!=\lambda!_{(s)}$. Applying this to \eqref{eqn:Intermediate}, we have
\begin{equation}
\frac{\det\!\big(f_{j}(z+y_{i})\big)_{i,j=1}^{s}}{\Delta(\mathbf{y})}
=
(-1)^{\frac{s(s-1)}{2}}\sum_{\substack{\lambda\\ \ell(\lambda)\le s}}
\frac{\det\!\big(f_{j}^{(\alpha_{i}(\lambda))}(z)\big)_{i,j=1}^{s}}{\lambda!_{(s)}}
\,s_{\lambda}(\mathbf{y}).
\end{equation}
Given the generating function of interest takes derivatives and coalesces variables at zero, it is natural to write this operation in the basis of monomial symmetric functions. In particular, we use the monomial expansion of the Schur polynomial,
\begin{equation}
s_{\lambda}(\mathbf{y})
=
\sum_{\eta\vdash |\lambda|} K_{\lambda\eta}\,m_{\eta}(\mathbf{y}),
\end{equation}
where $m_{\eta}$ is the monomial symmetric function. Hence
\begin{equation}
\left.\mathcal{D}^{\mu}_{\mathbf{y}}s_{\lambda}(\mathbf{y})\right|_{\mathbf{y}=0}
=
\mu!\,K_{\lambda\mu}.
\end{equation}
Substituting this into the Schur expansion proves the result.
\end{proof}
We now state the following theorem for finite-$N$, which provides a generalisation of \cite[Theorem 1.5]{SW26}.
\begin{Theorem}\label{thm:finiteNmicro}
Let $\mu,\nu$ be lists of non-negative integers with $s$ entries, which parametrise the chosen moment. Then 
\begin{equation}
M_{\mu,\nu}(z,N)
=
\mu!\nu!
\sum_{\substack{\lambda\vdash |\mu|,\ \ell(\lambda)\le s\\
\rho\vdash |\nu|,\ \ell(\rho)\le s}}
\frac{K_{\lambda\mu}K_{\rho\nu}}{\lambda!_{(s)}\,\rho!_{(s)}}
z^{|\rho|-|\lambda|}
\det\!\Big(\big(w^{\alpha_i(\lambda)}K_{N,s}^{(\alpha_i(\lambda))}(w)\big)^{(\beta_j(\rho))}\Big)_{i,j=1}^s,
\end{equation}
where $w=|z|^2$ and
\begin{equation}
K_{N,s}(w)\coloneq \sum_{m=0}^{N+s-1}w^{m}.
\end{equation}
\end{Theorem}
\begin{proof}
We begin by applying Lemma~\ref{lem:confluentkostka} twice to the exact expression in \eqref{AV}. For the first confluent step, fix $w_{1},\dots,w_{s}$ and define
\begin{equation}
f_{j}(u)\coloneq K_{N,s}(uw_{j}).
\end{equation}
Then
\begin{equation}
f_{j}^{(r)}(u)=w_{j}^{r}K_{N,s}^{(r)}(uw_{j}).
\end{equation}
Hence Lemma~\ref{lem:confluentkostka} yields
\begin{equation}
\begin{aligned}
&\left.
\mathcal{D}^{\mu}_{\mathbf{z}}
\frac{\det\!\big(K_{N,s}(z_{i}w_{j})\big)_{i,j=1}^{s}}{\Delta(\mathbf{z})}
\right|_{z_{1}=\cdots=z_{s}=z}
\\&=
(-1)^{\frac{s(s-1)}{2}}\mu!\sum_{\substack{\lambda\vdash |\mu|\\ \ell(\lambda)\le s}}
\frac{K_{\lambda\mu}}{\lambda!_{(s)}}
\det\!\big(w_{j}^{\alpha_{i}(\lambda)}K_{N,s}^{(\alpha_{i}(\lambda))}(zw_{j})\big)_{i,j=1}^{s}.
\end{aligned}
\end{equation}
After division by $\Delta(\mathbf{w})$, we apply Lemma \ref{lem:confluentkostka} again, now with
\begin{equation}\label{eqn:gkernel}
g_{i}(v)\coloneq v^{\alpha_{i}(\lambda)}K_{N,s}^{(\alpha_{i}(\lambda))}(zv).
\end{equation}
We obtain
\begin{equation}
M_{\mu,\nu}(z,N)
=
\mu!\nu!
\sum_{\substack{\lambda\vdash |\mu|,\ \ell(\lambda)\le s\\
\rho\vdash |\nu|,\ \ell(\rho)\le s}}
\frac{K_{\lambda\mu}K_{\rho\nu}}{\lambda!_{(s)}\,\rho!_{(s)}}
\det\!\big(g_{i}^{(\beta_{j}(\rho))}(\overline{z})\big)_{i,j=1}^{s}.
\end{equation}
In order to simplify evaluating $g_{i}^{(\beta_{j}(\rho))}(\overline{z})$, it is useful to introduce
\begin{equation}
H_i(u)\coloneq u^{\alpha_i(\lambda)}K_{N,s}^{(\alpha_i(\lambda))}(u).
\end{equation}
In particular, it trivialises the derivatives
\begin{equation}
\begin{aligned}
g_{i}^{(\beta_{j}(\rho))}(\overline{z}) &=\frac{\partial^{\beta_j(\rho)}}{\partial \overline{z}^{\beta_j(\rho)}} \left( \overline{z}^{\alpha_i(\lambda)}\,K_{N,s}^{(\alpha_i(\lambda))}(|z|^2) \right)\\ &=\frac{\partial^{\beta_j(\rho)}}{\partial \overline{z}^{\beta_j(\rho)}} \left(z^{-\alpha_i(\lambda)}\,H_i(|z|^2) \right)
\\ 
&=z^{\beta_j(\rho)-\alpha_i(\lambda)}H_i^{(\beta_j(\rho))}(|z|^2),
\end{aligned}      
\end{equation}
where the second line follows from the definition of $H_i(z\overline{z})$ and the third line from the chain rule. Therefore, upon factoring terms from each row and column of the determinant, we have
\begin{equation}\label{eqn:someequation}
\det\Big(g_{i}^{(\beta_{j}(\rho))}(\overline{z})\Big)_{i,j=1}^s
=
z^{|\rho|-|\lambda|}
\det\!\Big(\big(w^{\alpha_i(\lambda)}K_{N,s}^{(\alpha_i(\lambda))}(w)\big)^{(\beta_j(\rho))}\Big)_{i,j=1}^s.
\end{equation}
Using \eqref{eqn:gkernel} it is straightforward to see the determinant \eqref{eqn:someequation} is regular at $z=0$. This completes the proof.
\end{proof}
\begin{Corollary}[\cite{SW26}, Theorem 1.5]
For any $z\in\mathbb{C}$ and any positive integers $N$ and $s$, one has the exact identity
\begin{equation}
\mathbb{E}\big[|\Lambda_N'(z)|^{2s}\big]
=
\sum_{\lambda,\rho\in Y_s}
\frac{f_\lambda f_\rho}{\lambda!\rho!}
\det\left\{
\big(w^{\lambda_i+s-i}K_{N,s}^{(\lambda_i+s-i)}(w)\big)^{(\rho_j+s-j)}
\right\}_{i,j=1}^s,
\end{equation}
where $w=|z|^2$, $K_{N,s}(w)=\sum_{j=0}^{N+s-1}w^j$, and  $Y_s$ denotes the set of partitions of $s$.
\end{Corollary}
\begin{proof}
Starting from Theorem \ref{thm:finiteNmicro}, we set $\mu=\nu=(1^s)$ and observe that in this case the Kostka numbers reduce to $K_{\lambda,(1^s)}=f_\lambda$, where $f_{\lambda}$ is the number of standard tableaux of shape $\lambda$. We have $|\lambda|=|\rho|=s$, so the factor $z^{|\rho|-|\lambda|}$ is equal to $1$ and we have $\mu!=\nu!=1$.
\end{proof}

\begin{proof}[Proof of Theorem \ref{thm:micro2}]
Using \cite[Equation (5.25)]{SW26} we have
\begin{equation}
(u^{a}K^{(a)}_{N,s}(u))^{(b)}\sim N^{a+b+1}\int_0^1 x^{a+b}e^{-\tau x}\,dx.
\end{equation}
Hence, denoting $z_N=1-\frac{c}{N}$ and $w_N=|z_N|^2$,
\begin{equation}
\begin{split}
z_N^{|\rho|-|\lambda|}
\det\!&\Big(\big(u^{\alpha_i(\lambda)}K_{N,s}^{(\alpha_i(\lambda))}(u)\big)^{(\beta_j(\rho))}\Big)_{i,j=1}^s\Bigg|_{u = w_N}\\
&\sim N^{|\mu|+|\nu|+s^{2}}
\left(
\det\!\big(I_{\alpha_{i}(\lambda)+\beta_{j}(\rho)}(\tau)\big)_{i,j=1}^{s}
\right). \label{det-asymptotics}
\end{split}
\end{equation}
The result \eqref{climit} follows by substituting \eqref{det-asymptotics} into the finite-$N$ formula given in Theorem~\ref{thm:finiteNmicro}, while \eqref{c0limit} follows from \eqref{climit} and $I_{r}(0)=\frac{1}{r+1}$. Finally, we verify the strict positivity. By Andr\'eief's identity, the right-hand side of \eqref{climit} is
\begin{equation}
\frac{\mu!\nu!}{s!}\int_{[0,1]^{s}}\left(\prod_{j=1}^{s}dx_{j}\,e^{-\tau x_{j}}\right)\sum_{\lambda \vdash |\mu|, \ell(\lambda) \leq s}\frac{K_{\lambda \mu}}{\lambda!_{(s)}}s_{\lambda}(\bm x)\sum_{\rho \vdash |\nu|, \ell(\rho) \leq s}\frac{K_{\rho \nu}}{\rho!_{(s)}}s_{\rho}(\bm x)\Delta^{2}(\bm x).\label{sumschur}
\end{equation}
By Schur positivity and non-negativity of the Kostka numbers, we conclude that the integrand of \eqref{sumschur} is positive on a set of full measure and this verifies the strict positivity property.
\end{proof}

\section[Moments of higher order derivatives of the Riemann zeta function with sigma > 1/2]{Moments of higher order derivatives of the Riemann zeta function with \texorpdfstring{$\sigma > \frac{1}{2}$}{sigma > 1/2}}\label{section3}
The goal of this section is to establish Proposition \ref{thm:ZetaThmLindel} and Theorem \ref{thm:ZetaThms2}. We follow the approach described in the book by Titchmarsh \cite[Chapter 7]{titchmarsh} and \cite[Section 3]{SW26}. We recall that $\zeta(s)$ can be represented by its Dirichlet series or Euler product as
\begin{equation}
	\zeta(s) = \sum_{n=1}^{\infty}\frac{1}{n^{s}} = \prod_{p}\frac{1}{1-p^{-s}}, \qquad \mathrm{Re}(s) > 1. \label{zetadef}
\end{equation}
 When we take powers and products of $\zeta$ it will be useful to recall the identity, for sufficiently small complex numbers $\bm{\alpha} = (\alpha_1,\ldots,\alpha_K) \in \mathbb{C}^{K}$,
\begin{equation}
\prod_{j=1}^{K}\zeta(s+\alpha_j) = \sum_{m=1}^{\infty}\frac{\sigma_{\bm \alpha}(m)}{m^{s}}, \qquad \sigma_{\bm \alpha}(m) = \sum_{n_{1}\ldots n_{K}=m}n_{1}^{-\alpha_1}\ldots n_{K}^{-\alpha_{K}}, \label{prodzet}
\end{equation}
which follows from the Dirichlet series in \eqref{zetadef}. By differentiating the shifts, we immediately get
\begin{equation}
\prod_{j=1}^{K}\zeta^{(\mu_j)}(s) = (-1)^{|\mu|}\sum_{m=1}^{\infty}\frac{A_{\mu}(m)}{m^{s}}, \label{pro}
\end{equation}
where
\begin{equation}
	A_{\mu}(m) := (-1)^{|\mu|}\partial^{\mu}_{\bm \alpha}\sigma_{\bm \alpha}(m)|_{\bm \alpha = 0} = \sum_{n_1 \ldots n_{K}=m}\log(n_1)^{\mu_1}\ldots \log(n_K)^{\mu_K}. \label{defamu}
\end{equation}
Now we discuss the Lindel\"of hypothesis (LH) and its relation to the mean value \eqref{eqn:ZetaMoments}. Recall that LH is a conjectured bound on $\zeta$ of the form $\zeta(\sigma+it) = O(t^{\epsilon})$ as $t \to \infty$ for $\sigma > \frac{1}{2}$ and any $\epsilon>0$. By analyticity and applying Cauchy's theorem we get the same bound for any order of the derivative, i.e.\ $\zeta^{(m)}(\sigma+it) = O(t^{\epsilon})$ as $t \to \infty$ for any $m \in \mathbb{N}$. As explained in \cite{titchmarsh}, LH is closely related to mean values of the Riemann zeta function off the critical line. In particular, \cite[Chapter 13]{titchmarsh} shows that LH is equivalent to the mean value
\begin{equation}
	\lim_{T \to \infty}\frac{1}{T}\int_{1}^{T}|\zeta(\sigma+it)|^{2K}dt = \sum_{n=1}^{\infty}\frac{d^{2}_{K}(n)}{n^{2\sigma}}, \qquad \sigma > \frac{1}{2},
\end{equation}
holding for all $K \in \mathbb{N}$. Here, $d_{K}(n)$ is the $K^{\mathrm{th}}$ divisor function of $n$. This result is known for $K=1,2$, but not for $K>2$ without additional restrictions on $\sigma$.
\begin{lemma}\label{lem:LogConvolution}
	 Assuming the Lindel\"of hypothesis, for any $\sigma>\frac{1}{2}$,
	\begin{equation}
		\lim_{T\rightarrow \infty} \frac{1}{T} \int_{1}^{T}
		\prod_{i=1}^{K}\zeta^{(\mu_i)}(\sigma+it)
		\prod_{j=1}^{L} \overline{\zeta^{(\nu_j)}(\sigma+it)}dt=(-1)^{|\mu|+|\nu|}\sum_{m=1}^{\infty} \frac{A_{\mu}(m)A_{\nu}(m)}{m^{2\sigma}}. \label{al}
	\end{equation} 
\end{lemma}

\begin{proof}
This follows identically to \cite[Proof of Theorem 7.9]{titchmarsh} where we replace $\zeta^{K}(s)$ in the proof with $\prod_{j=1}^{K}\zeta^{(\mu_j)}(s)$ or $\prod_{k=1}^{L}\zeta^{(\nu_k)}(s)$, using the Dirichlet series \eqref{pro} and the bound $|A_{\mu}(m)| \leq \log(m)^{|\mu|}d_{K}(m)$.
\end{proof}
The next step is to compute the limit as $\sigma \to \frac{1}{2}$ on the right-hand side of \eqref{al}. 
\begin{lemma}
	\label{lem:zetid}
We have the identity
\begin{equation}\label{factorH}
\begin{aligned}(-1)^{|\mu|+|\nu|}&\sum_{m=1}^{\infty}\frac{A_{\mu}(m)A_{\nu}(m)}{m^{2\sigma}} \\&= \partial^{\mu}_{\bm \alpha}\partial^{\nu}_{\bm \beta}\left(H_{\mu,\nu}(\bm \alpha, \bm \beta,\sigma)\prod_{i=1}^{K}\prod_{j=1}^{L}\zeta(2\sigma+\alpha_{i}+\beta_{j})\right)\bigg{|}_{\bm \alpha=\bm \beta=0}, 
\end{aligned}
\end{equation}
where $H_{\mu,\nu}(\bm \alpha, \bm \beta, \sigma)$ is an analytic function of its parameters for any $\sigma > \frac{1}{4}$ and $\bm \alpha$, $\bm \beta$ in a neighbourhood of zero. We also have
\begin{equation}
\lim_{\sigma \to \frac{1}{2}}H_{\mu,\nu}(\bm 0, \bm 0,\sigma) = a_{\mu,\nu}, \label{limH}
\end{equation}
where $a_{\mu,\nu}$ is the arithmetic constant \eqref{eqn:generalarithmetic}.
\end{lemma}
\begin{proof}
To prove \eqref{factorH}, note that by definition \eqref{defamu}, it suffices to show
\begin{equation}
\sum_{m=1}^{\infty}\frac{\sigma_{\bm \alpha}(m)\sigma_{\bm \beta}(m)}{m^{2\sigma}} = H_{\mu,\nu}(\bm \alpha, \bm \beta,\sigma)\prod_{i=1}^{K}\prod_{j=1}^{L}\zeta(2\sigma+\alpha_{i}+\beta_{j}). \label{suff}
\end{equation}
Using the Euler product for Dirichlet series, the left-hand side of \eqref{suff} is
\begin{equation}\label{eqn:HDefinition}
	\begin{aligned}
		&\prod_{p}\left(1+\sum_{n=1}^{\infty}\frac{1}{p^{2n\sigma}}\left(\sum_{n_1 \dots n_{K}=p^n}n_1^{-\alpha_1} \dots n_{K}^{-\alpha_{K}}\right)\left(\sum_{m_1\dots m_{L}=p^n} m_1^{-\beta_1}\dots m_{L}^{-\beta_{L}}\right)
		\right)\\
		&=\prod_{p} \left( 1 + \sum_{n=1}^{\infty} \frac{1}{p^{2n\sigma}} 
		\sum_{\substack{a_1+\dots+a_{K}=n \\ b_1+\dots+b_{L}=n \\ a_i, b_j \geq 0}} 
		\frac{1}{p^{a_1 \alpha_1} \dots p^{a_{K} \alpha_{K}} p^{b_1 \beta_1} \dots p^{b_{L} \beta_{L}}} \right).
	\end{aligned}
\end{equation}
Consider the expansion,
\begin{equation}
	\begin{aligned}
		&\prod_{i=1}^{K} \prod_{j=1}^{L} \zeta(2\sigma +\alpha_i+\beta_j)^{-1}=\prod_p \exp\left(-\sum_{n=1}^{\infty}\sum_{i=1}^{K} \sum_{j=1}^{L} 
		\frac{1}{n}p^{-n(2\sigma +\alpha_i +\beta_j)} \right)
		\\ &=\prod_p\left(1-\sum_{i=1}^{K} \sum_{j=1}^{L}p^{-2\sigma-\alpha_i-\beta_j}+O\left(p^{-4\sigma-4\min_{i,j}\{\Re(\alpha_i),\Re(\beta_j) \}}\right) \right). \label{expanzet}
	\end{aligned}
\end{equation}
By definition, $H_{\mu,\nu}(\bm \alpha, \bm \beta,\sigma)$ is obtained by multiplying \eqref{eqn:HDefinition} and \eqref{expanzet}, so the terms of order $p^{-2\sigma-\alpha_i-\beta_j}$ in \eqref{expanzet} cancel with the $n=1$ term in \eqref{eqn:HDefinition}. This shows that $H_{\mu,\nu}(\bm \alpha, \bm \beta,\sigma)$ is analytic in all its variables for $\sigma > \frac{1}{4}$ and $\bm \alpha$, $\bm \beta$ in a neighbourhood of zero. Regarding the limit \eqref{limH}, we use the Euler product \eqref{zetadef} to write
\begin{equation}
	\lim_{\sigma \to \frac{1}{2}}H_{\mu,\nu}(\bm 0, \bm 0, \sigma) = \prod_{p}\left(1-p^{-1}\right)^{KL}\left(\sum_{m=0}^{\infty}\frac{d_{K}(p^{m})d_{L}(p^{m})}{p^{m}}\right).
\end{equation}
Then \eqref{limH} follows after recalling the identity $d_{K}(p^{m}) = \frac{(K+m-1)!}{m!(K-1)!}$.
\end{proof}
Now we can prove the conditional part of Proposition \ref{thm:ZetaThmLindel}.
\begin{proof}[Proof of \eqref{eqn:ZetaMoments} assuming LH]
We insert the result of Lemma \ref{lem:zetid} into Lemma \ref{lem:LogConvolution}. The leading order contribution as $\sigma \to \frac{1}{2}$ comes from applying the derivatives only to the product of terms of the form $\zeta(2\sigma+\alpha_{i}+\beta_{j})$, due to analyticity of $H_{\mu,\nu}(\bm \alpha, \bm \beta, \sigma)$. Hence, the left-hand side of \eqref{eqn:ZetaMoments} is asymptotic to 
\begin{equation}
a_{\mu,\nu}\mathcal{D}^{\nu}_{\bm{\beta}}\mathcal{D}^{\mu}_{\bm{\alpha}}\left(\prod_{i=1}^{K}\prod_{j=1}^{L} \zeta(2\sigma + \alpha_i + \beta_j)
			\right)\Bigg|_{\bm \alpha = \bm \beta = \bm 0}, \label{eqn:ZetaLogic}
\end{equation}
as $\sigma \to \frac{1}{2}$. We now compute the above coefficient. Taylor expanding near the origin, for $\sigma>\frac{1}{2}$ we have 
	\begin{equation}
		\zeta(2\sigma + \alpha_i+\beta_j) = \sum_{Q_{ij}=0}^{\infty}\sum_{R_{ij}=0}^{\infty}\frac{\zeta^{(R_{ij}+Q_{ij})}(2\sigma)}{R_{ij}!Q_{ij}!}\alpha_{i}^{Q_{ij}}\beta_{j}^{R_{ij}}. \label{zetataylor}
	\end{equation}
	Inserting \eqref{zetataylor} into the left-hand side of \eqref{eqn:ZetaLogic}, the derivatives at $0$ fix the row and column sums of $Q$ and $R$, leading to
	\begin{equation}\label{eqn:algebra1}
		\mathcal{D}^{\nu}_{\bm{\beta}}\mathcal{D}^{\mu}_{\bm{\alpha}}\left(\prod_{i=1}^{K}\prod_{j=1}^{L} \zeta(2\sigma + \alpha_i + \beta_j) 
		\right)\Bigg|_{\bm \alpha = \bm \beta = \bm 0} =
		\mu!\nu!\sum_{\substack{Q \in \mathcal{M}_{\mu,\cdot} \\ R \in \mathcal{M}_{\cdot,\nu}}} \prod_{i=1}^{K}  \prod_{j=1}^{L}\frac{\zeta^{(R_{ij} + Q_{ij})}(2\sigma)}{R_{ij}! Q_{ij}!}.
	\end{equation}
	The result then follows from inserting the following into \eqref{eqn:algebra1},
	\begin{equation}
		\zeta^{(k)}(2\sigma)\sim \frac{(-1)^kk!}{(2\sigma-1)^{k+1}} \quad \text{as} \quad \sigma\rightarrow \frac{1}{2}. \label{zet1}
	\end{equation}
\end{proof}
Finally, we complete the proof of the unconditional part of Proposition \ref{thm:ZetaThmLindel} by taking the limit $\sigma \to \frac{1}{2}$ in the next Lemma. This also completes the proof of Theorem \ref{thm:ZetaThms2} by the same calculations given in the proof of Lemma \ref{lem:len2} combined with \eqref{eqn:algebra1} and \eqref{zet1}.
\begin{lemma}\label{lemma2}
Let $\mu,\nu$ be lists of non-negative integers with $K=\ell(\mu)\leq 2$ and $L=\ell(\nu)\leq 2$. Then, for $\alpha_1,\alpha_2,\beta_1,\beta_2$ sufficiently small complex numbers and $\sigma>\frac{1}{2}$ fixed, we have:
\begin{equation}\label{eqn:secondmomentgen}
\lim_{T\rightarrow \infty} \frac{1}{T} \int_{1}^{T}
\prod_{i=1}^{K}\zeta^{(\mu_i)}(\sigma+it)
\prod_{j=1}^{L} \overline{\zeta^{(\nu_j)}(\sigma+it)}\,dt =\mathcal{D}^{\nu}_{\bm{\beta}}\mathcal{D}^{\mu}_{\bm{\alpha}} \; \mathcal{G}(\bm{\alpha},\bm{\beta},\sigma)|_{\bm \alpha=\bm \beta = 0},
\end{equation}
where
\begin{equation}\label{eq:Kmunu}
\mathcal{G}(\bm{\alpha},\bm{\beta},\sigma)
=
\begin{cases}
\displaystyle \zeta(2\sigma+\alpha_1+\beta_1),
& K=1,\; L=1,\\[4mm]
\displaystyle \zeta(2\sigma+\alpha_1+\beta_1)\zeta(2\sigma+\alpha_1+\beta_2),
& K=1,\; L=2,\\[4mm]
\displaystyle \zeta(2\sigma+\alpha_1+\beta_1)\zeta(2\sigma+\alpha_2+\beta_1),
& K=2,\; L=1,\\[4mm]
\displaystyle \frac{\prod_{i,j=1}^2\zeta(2\sigma+\alpha_i+\beta_j)}
{\zeta(4\sigma+\alpha_1+\alpha_2+\beta_1+\beta_2)},
& K=2,\; L=2.
\end{cases}
\end{equation}
\end{lemma}
\begin{proof}
The case $K=L=1$ follows
directly from the mean-value theorem for Dirichlet series \cite[Theorem 7.1]{titchmarsh} that for $\sigma>\frac{1}{2}$,
\begin{equation}
\lim_{T \to \infty} \frac{1}{T} \int_1^T \zeta^{(\mu_1)}(\sigma + it) \overline{\zeta^{(\nu_1)}(\sigma + it)}dt =
\sum_{n=1}^{\infty}\frac{\left(-\log(n)\right)^{\mu_1}\left(-\log(n)\right)^{\nu_1}}{n^{2\sigma}}.
\end{equation}
Hence for $\alpha_1, \beta_1$ sufficiently small complex numbers and $\mu_1,\nu_1\in \mathbb{Z}_{\geq 0}$ we see
\begin{equation}
\begin{aligned}
\lim_{T \to \infty} \frac{1}{T} \int_1^T \zeta^{(\mu_1)}(\sigma + it) \overline{\zeta^{(\nu_1)}(\sigma + it)}dt  =
\frac{\partial^{\mu_1}}{\partial \alpha_1^{\mu_1}} \frac{\partial^{\nu_1}}{\partial \beta_1^{\nu_1}} 
\zeta(2\sigma+\alpha_1+\beta_1)
\Bigg
|_{\alpha_1=\beta_1=0},
\end{aligned}
\end{equation}
as required. Without loss of generality, suppose $K=1$ and $L=2$.  In particular, consider the two Dirichlet series:
\begin{equation}
\begin{aligned}
&\zeta^{(\mu_1)}(\sigma + it) = \sum_{n=1}^{\infty} \frac{(-\log(n))^{\mu_1}}{n^{\sigma+it}},
\\
&\overline{\zeta^{(\nu_1)}(\sigma + it)\zeta^{(\nu_2)}(\sigma + it)}= \sum_{n=1}^{\infty} \frac{1}{n^{\overline{s}}}\left(\sum_{m_1m_2=n} (-\log(m_1))^{\nu_1}(-\log(m_2))^{\nu_2}\right).
\end{aligned}
\end{equation}
Using the mean-value theorem for Dirichlet series, and taking $\Re(2\sigma +\alpha_i +\beta_j)>1$ , we conclude that: 
\begin{equation}
\begin{aligned}
&\lim_{T \to \infty} \frac{1}{T} \int_1^T \zeta^{(\mu_1)}(\sigma + it) \overline{\zeta^{(\nu_1)}(\sigma + it)}\overline{\zeta^{(\nu_2)}(\sigma + it)}dt\\
&=\frac{\partial^{\mu_1}}{\partial \alpha_1^{\mu_1}}\frac{\partial^{\nu_1}}{\partial \beta_1^{\nu_1}}\frac{\partial^{\nu_2}}{\partial \beta_2^{\nu_2}} \left(\sum_{n =1}^\infty \frac{1}{n^{2\sigma+\alpha_1}} \left(  \sum_{m_1m_2=n} \frac{1}{m_1^{\beta_1} m_2^{\beta_2}}\right) \right)\Bigg
|_{\bm{\alpha}=\bm{\beta}=0}
\\ 
&=\mathcal{D}_{\bm{\beta}}^\nu  \,\mathcal{D}_{\bm{\alpha}}^\mu  \, \left(
\sum_{m_1,m_2=1}^{\infty} \frac{1}{m_1^{2\sigma+\alpha_1+\beta_1}m_2^{2\sigma+\alpha_1+\beta_2}}
\right)\Bigg
|_{\bm{\alpha}=\bm{\beta}= 0}.
\end{aligned}
\end{equation}
The result follows by writing the above Dirichlet series in terms of the zeta function. Finally, consider the case $K=L=2$. We follow the approach outlined in  \cite[Theorem 7.5]{titchmarsh} and \cite[Lemma 3.1]{SW26}. We start with the approximate functional equation, 
\begin{equation}
\zeta(s)
= \sum_{n \leq x} n^{-s}+
\chi(s)\sum_{n \leq y} n^{s - 1}
+
O\bigl(x^{-\Re(s)}+ |t|^{\frac{1}{2} - \Re(s)}y^{\Re(s)-1}\bigr).
\end{equation}
where $\chi(s) = 2^{s}\pi^{s-1}\sin(\pi s/2)\Gamma(1-s)$. For $s=\sigma+it$, this holds uniformly for $x,y \geq 1$ and $xy=\frac{t}{2\pi}$ and $0 <\Re(s)< 1$. As seen in \cite{conrey4}, using Cauchy's theorem we have,
\begin{equation}
\begin{aligned}
\frac{d^k}{ds^k}\zeta(s)
&= \frac{d^k}{ds^k}\left(\sum_{n \leq x} n^{-s}\right)+
\frac{d^k}{ds^k}\left(\chi(s)\sum_{n \leq y} n^{s - 1}\right)
\\
&+O\left(\left(x^{-\Re(s)}+ |t|^{\frac{1}{2} - \Re(s)}y^{\Re(s)-1}\right)\log(|t|)^k\right)
\\
&=Z_1 +Z_2+O(t^{-\Re(s)/2}\log(|t|)^{k}),
\end{aligned}
\end{equation}
where we set $x=y=\sqrt{\frac{t}{2\pi}}$. In order to evaluate moments of the form (\ref{eqn:secondmomentgen}), we denote the relevant terms,
\begin{equation}
	\begin{aligned}
		\begin{aligned}
			\frac{d^{\mu_1}}{ds^{\mu_1}}\zeta(s) &= A_1 + A_2 + R_A, \qquad &
			\frac{d^{\nu_1}}{ds^{\nu_1}}\overline{\zeta(s)} &= \overline{C_1} + \overline{C_2} + \overline{R_C}, \\
			\frac{d^{\mu_2}}{ds^{\mu_2}}\zeta(s) &= B_1 + B_2 + R_B, \qquad &
			\frac{d^{\nu_2}}{ds^{\nu_2}}\overline{\zeta(s)} &= \overline{D_1} + \overline{D_2} + \overline{R_D},
		\end{aligned}
	\end{aligned}
\end{equation}
where the contributions $R$ denote the respective remainders. The contribution from the leading term is
	\begin{equation}\label{eqn:Z1Asympt}
		\partial^{\mu}_{\bm \alpha}\partial^{\nu}_{\bm \beta}\frac{1}{T}\int_{1}^T\sum_{n_1,n_2,n_3,n_4\leq \sqrt{\frac{t}{2\pi}}}n_1^{-s-\alpha_1}n_2^{-s-\alpha_2}n_3^{-\overline{s}-\beta_1}n_4^{-\overline{s}-\beta_2}\Bigg|_{\bm{\alpha}=\bm{\beta}=0}.
	\end{equation}
Taking the diagonal contribution $n_{1}n_{2}=n_{3}n_{4}$, we obtain the right-hand side given by \eqref{eq:Kmunu} due to the following identity
\begin{equation}
	\begin{split}
&\sum_{n_1 n_{2} = n_{3}n_{4}}n_1^{-s-\alpha_1}n_2^{-s-\alpha_2}n_3^{-\overline{s}-\beta_1}n_4^{-\overline{s}-\beta_2}=\frac{\prod_{i,j=1}^{2}\zeta(2\sigma+\alpha_i+\beta_j)}{\zeta(4\sigma+\alpha_1+\alpha_{2}+\beta_{1}+\beta_{2})},
\end{split}
\end{equation}
which is proved in \cite[Section 3]{SW26} by expanding the left-hand side using an Euler product and summing the resulting geometric series. The off-diagonal terms are bounded in the same way as \cite[Theorem 7.5]{titchmarsh} and furthermore the effect of differentiating the error terms that arise there does not change their magnitude except for logarithmic factors.
\end{proof}

\bibliographystyle{plain}
\bibliography{bibliography}

\end{document}